\documentclass[runningheads]{llncs}

\usepackage{tikz,tikz-qtree,tikz-qtree-compat}
\usetikzlibrary{shapes,decorations,arrows,calc,fit}
\usepackage{graphicx}
\usepackage{hyperref}

\usepackage{amsfonts,amssymb}
\usepackage{amsmath}
\usepackage{bbm}
\usepackage{dsfont}
\usepackage[bb=dsserif]{mathalpha}
\usepackage{bm}

\usepackage{color,xspace,colortbl}
\usepackage{float}
\usepackage[ruled]{algorithm2e}
\usepackage[noend]{algpseudocode}
\usepackage{enumitem}

\newtheorem{observation}[theorem]{Observation}
\spnewtheorem{sublemma}[theorem]{Sublemma}{\bfseries}{\itshape}

\newcommand{\flow}{{\sf excess}}

\newcommand{\incomp}{\perp}

\newcommand{\CC}{V}

\newcommand{\DC}{\mathrm{DC}}

\usepackage{todonotes}

\colorlet{LightViolet}{violet!40}
\colorlet{LightRed}{red!40}
\colorlet{LightOrange}{orange!40}
\colorlet{LightGreen}{green!40}
\colorlet{LightBlue}{blue!40}
\colorlet{DarkGreen}{green!50!black}
\colorlet{DarkRed}{red!70!black}
\colorlet{DarkCyan}{red!70!black}
\colorlet{DarkBlue}{blue!80!black}
\definecolor{DarkOrange}{rgb}{1.0, 0.49, 0.0}
\definecolor{Airforceblue}{rgb}{0.36, 0.54, 0.66}

\newcommand{\nop}[1]{}

\begin{document}

\title{Optimizing Entropic and Polymatroid Functions Subject to Difference Constraints\thanks{S. Im was supported in part by NSF grants CCF-1617653, CCF-1844939 and CCF-2121745. B. Moseley was supported in part by by NSF grants CCF-1824303,  CCF-1845146, CCF-2121744 and CMMI-1938909, a Google Research Award, an Infor Research Award, and a Carnegie Bosch Junior Faculty Chair. K. Pruhs was supported in part  by NSF grants  CCF-1907673,  CCF-2036077, CCF-2209654 and an IBM Faculty Award. A. Samadian contributed to this work while he was a PhD student at the University of Pittsburgh and is now affiliated with Google Pittsburgh.}}
\titlerunning{Optimizing Polymatroid Functions}
%
\author{Sungjin Im\inst{1} \and
Benjamin Moseley\inst{2} \and
Hung Q. Ngo\inst{3} \and
Kirk Pruhs\inst{4}
\and
Alireza Samadian\inst{4}}
\authorrunning{S. Im et. al.}
%
\institute{University of California, Merced CA, 95343 \and
Tepper School of Business, Carnegie Mellon University, Pittsburgh PA, 15213
\and
RelationalAI, Inc., Berkeley, CA, 94704.
\and
Computer Science Department, University of Pittsburgh, Pittsburgh PA, 15260
}
\maketitle              
\begin{abstract}
We consider a class of optimization problems
that involve determining the maximum value that
a function in a particular
class  can attain subject to
a collection of difference constraints.
We show that a particular linear programming technique, based on duality and projections,
can be used to rederive some structural results that were previously established using
more ad hoc methods. We then show that this technique can be used to obtain a polynomial-time
algorithm for a certain type of simple difference constraints.
Finally we give lower bound results that show that certain possible extensions of these results
are probably not feasible.

\end{abstract}

\section{Introduction}

We consider a class of optimization problems
that involve determining the maximum value that
a function in a particular
class $\mathcal C$ can attain subject to
a collection of difference constraints.
 So written as a mathematical program, these problems are of the form $DC[{\mathcal C}]$:
 \begin{equation}
\begin{array}{rrclcl}
\displaystyle \max & \multicolumn{3}{l}{h([n]) - h( \emptyset)} \\
\textrm{s.t.} & h(Y_i ) - h( X_i) & \leq & c_i  \qquad i \in [k] \\
&\displaystyle h & \in & \mathcal C
\end{array}  \nonumber
\end{equation}
where each $X_i \subsetneq  Y_i \subseteq [n]$,
and the ``variable'' is the function (or vector) $h:2^n \rightarrow \Re^+$.

We are primarily interested in two classes $\mathcal C$ of functions.
The first is the class
$\Gamma_n^*$ of entropic functions.
If $h$ is an entropic function
then  the difference $ h(Y_i) - h(X_i)$ is equal to $h(Y_i \mid X_i)$,
the conditional entropy of $Y_i$ conditioned on the knowledge
of $X_i$.
Thus $DC[\Gamma^*]$ is the problem of determining the  maximum possible
entropy subject to  conditional entropy constraints.
The second is the class $\Gamma_n$ of polymatroid functions.
Other classes of functions that will play a role in our story are
 the class $N_n$ of weighted coverage  functions (normal functions in the database literature), and the class
$M_n$ of modular functions.

Observe that for all collections $DC$ of difference constraints,
$$DC[M_n] \le DC[N_n] \le DC[\Gamma_n^* ] \le DC[\Gamma_n]$$
where here $DC[{\mathcal C}]$ should be interpreted to mean the optimal objective value.
The first inequality follows because modular functions are weighted coverage functions,
the second inequality follows because weighted coverage functions
are entropic functions, and last inequality follows because
entropic functions are polymatroid functions.
It will be convenient to refer to these quantities as
the modular bound, the coverage bound, the entropic bound
and the polymatroid bound, respectively. It is known that for arbitrary
difference constraints that the gaps between these bounds can arbitrarily large~\cite{panda}.

\subsection{Database Applications and Background}

Our main motivation for considering these types of optimization problems arises
from applications in databases, for example in
bounding the cardinality  of
a table that is formed by joining smaller tables~\cite{NGOICDT2022}.
The difference constraints express knowledge
about the data
that might come from some external understanding of the relation
between various attributes, or information that is easily
obtained from the data.
So for example, consider a simple database
for  US postal service that contains the following three attributes (among others),
attribute 1 is zip code, attribute 2 is city and
attribute 3 is state. Assume that it is known that
that there are at most 50 states, that
no state contains more than 2,598 zip codes,
that no city is in more than one state, and no
state contains more than 9,197 cities.
Then suppose one joins together a collection of tables to get a table $T$
consisting of (zip code, city, state) tuples.
Then the number of tuples in $T$ can be upper bounded using the following instance of
$DC[\Gamma^*_n]$ problem:
 \begin{equation}
\begin{array}{rrclcl}
\displaystyle \max & \multicolumn{3}{l}{h(\{1, 2, 3 \} ) - h(\emptyset)} \\
\textrm{s.t.} & h(\{3\} \mid \emptyset) &\le& \lg 50  \\
&h(\{1, 3\} \mid \{3\}) &\le& \lg 2598  \\
&h(\{2, 3\} \mid \{2\})& \le& \lg 1  \\
&h(\{2, 3\} \mid \{3\})& \le& \lg 9197  \\
&\displaystyle h & \in & \Gamma_n^*
\end{array}
\nonumber
\end{equation}
In particular,  if this optimal objective is $E$ then $2^E$ is an upper
bound to the cardinality of the number of tuples in $T$ (essentially because uniform distributions have maximum entropy)~\cite{panda}. For certain common types of queries, e.g. disjunctive datalog queries, and joins with functional dependencies
and/or cardinality constraints, it is know that this entropic
bound is asymptotically tight~\cite{panda,csma}.
Unfortunately the space of entropic functions is a complicated.
For example,
if arbitrary linear constraints are allowed (instead of just difference constraints)
the problem is not even computable~\cite{Ting20219,arxiv.2206.03465}.
Thus the interest in the polymatroid bound in the database community derives from
the fact that it represents a clearly computable, and potentially even efficiently computable,
upper bound.

To date database researchers have largely considered certain
classes of difference constraints that commonly/naturally arise in database
applications~\cite{NGOICDT2022}. Of particular interest to us here are:
\begin{description}
\item[Acyclic Difference Instances:]
An acyclic instance is one where the dependency digraph of the difference constraints is acyclic.
The vertices of the dependency digraph is the universe $[n]$ and $(u, v)$ is a directed edge if
and only if there exists an $i \in [k]$  such that $u \in X_i$ and
$v \in Y_i - X_i$.
\item[Simple Difference Constraint:]
A simple difference constraint is one where $|X_i| \le 1$. A simple instance is one where all
difference constraints are simple.
\item[Cardinality Constraint:] A cardinality constraint is one where $|X_i| = 0$.
\item[Functional Dependency] A functional dependency is one where $c_i=0$.
\end{description}

In \cite{DBLP:conf/pods/000118} it was shown that, for acyclic instances  the modular bound is equal to the
polymatroid bound. The proof technique was to show that every polymatroid  function $h$ that
satisfies the difference constraints can be converted in a {\em modular} function $h'$ that also satisfies
the difference constraints and that has the same objective value as $h$. As the modular bound
can be computed by a polynomially sized linear program, this observation yields a polynomial
time algorithm to compute the entropic/polymatroid bound for acyclic instances.
In \cite{DBLP:conf/pods/KhamisK0S20} it was shown that, for simple instances  the coverage bound is equal to the
polymatroid bound. The proof technique was to show that every polymatroid function $h$
satisfying the difference constraints could be converted in a {\em weighted coverage function} $h'$ satisfying
the difference constraints and having the same objective value as $h$. This observation
doesn't immediately yield a polynomial time algorithm as the natural linear programming
formulation of  the coverage bound has exponentially many variables and exponentially many constraints.

There is a rich history of research in this area in the database theory community,
that we can't hope to do justice to here, but a good starting point for reader
are the surveys~\cite{NGOICDT2022,DBLP:conf/pods/000118}.
Also entropy implication problems are central to the area of information theory,
see for example the textbooks~\cite{InfoTheoryText,MR3183760}.

\subsection{Our Contributions}
\label{sec:contirbutions}

\subsubsection{The Dual-Project-Dual Technique}

We apply  a linear programming based technique,
which we will call the dual-project-dual  technique,
that allows us to rederive  the results that
$DC[M_n] = DC[\Gamma_n]$ for acyclic instances and $DC[N_n] = DC[\Gamma_n]$ for simple instances. We then extend these results in various ways.

Our starting point is a natural linear programming
formulation $P$ for the problem of computing $DC[\Gamma_n]$:
 \begin{equation}
\begin{array}{rrclcl}
\displaystyle P: \quad\quad  \max & \multicolumn{3}{l}{h([n]) - h(\emptyset)} \\
\textrm{s.t.} & h(Y\cup X) - h(X) - h(Y) +  h( Y \cap X) & \leq & 0 \qquad \forall X \forall Y  X \perp Y  \\
&h(Y) - h( X) & \geq & 0 \qquad\forall X \forall Y  X \subsetneq Y \\
 &\displaystyle h(Y_i) - h(X_i) &\le& c_i  \qquad \forall i \in [k]
\end{array}
\nonumber
\end{equation}
where there is a variable $h(X)$ for each subset $X$ of the universe  $[n]$,
and  $X \incomp Y$ means $X \not\subseteq Y$ and $Y \not\subseteq X$.
We will adopt the convention that all variables in our mathematical programs
are constrained to be nonnegative unless explicitly mentioned otherwise.
An optimal polymatroid function $h'$ would then be $h'(X) = h(X) - h(\emptyset)$,
where the values on the right hand side come from the linear program $P$.
Note that the linear program $P$ has both exponentially many variables
and exponentially many constraints. But critically this linear
program only has linearly many constraints  where the constant
on the right-hand-size of the constraint is nonzero
(one for each difference constraint).

The first step of our dual-project-dual technique is to take the
dual of $P$ to obtain a linear program $D$.
If we associate a dual variables $\sigma_{X,Y}$,
dual variables $\mu_{X,Y}$ and
dual variables $\delta_i$ with the
three types of constraints in $P$ (in that order), then the dual linear program $D$ is:

\begin{equation}
\begin{array}{rrclcl}
\displaystyle D: \quad \quad \min & \multicolumn{3}{l}{\sum_{i \in [k]} c_i \cdot \delta_i} \\
\textrm{s.t.} & \flow([n])  & \geq & 1  \\
&\flow(\emptyset)  & \geq & -1\\
&\flow(Z)              & \geq& 0, \qquad \forall Z  \ne \emptyset, [n]
\end{array}
\nonumber
\end{equation}
where $\flow(Z)$ is  defined as follows:
\begin{multline}
   \flow(Z) :=
     \sum_{i: Z = Y_i} \delta_i -
   \sum_{i: Z=X_i }\delta_i +
   \sum_{\substack{I\incomp J\\I\cap J = Z}}\sigma_{I,J}\\
   +
   \sum_{\substack{I'\incomp J'\\I'\cup J' = Z}}\sigma_{I',J'}
   - \sum_{J: J\incomp Z}\sigma_{Z,J}-
   \sum_{X: X\subset Z}\mu_{X,Z}+ \sum_{Y: Z\subset Y}\mu_{Z,Y}
   \label{eqn:flow}  
\end{multline}
See Figure \ref{fig:flow} for an illustration of $\flow(Z)$.
\begin{figure}[!ht]
\centering \begin{tikzpicture}[domain=0:20, 
every node/.style={font=\large}, scale=0.75, every node/.style={scale=0.7}]
\node[circle,draw] at (5,5) (Z) {$Z$};
\node[] at (3,8) (Y) {$Y$};
\node[] at (3,2) (X) {$X_i$};
\node[] at (5,7) (I) {$I$};
\node[] at (7,6) (J) {$J$};
\node[] at (7,8) (IJ) {$I\cup J$};
\node[] at (5,3) (I') {$I'$};
\node[] at (7,4) (J') {$J'$};
\node[] at (7,2) (IJ') {$I'\cap J'$};
\node[left = .2 of Z] (eq) {$I\cap J = I' \cup J' = Z$};  

\node[color=black!55!green] at (6,6.5) (sij) {$+\sigma_{I,J}$};
\node[color=black!55!green] at (6,3.5) (sij) {$+\sigma_{I',J'}$};

\draw[thick] (Z) -- (Y) node [midway,left=5pt,color=black!55!green] {$+\mu_{Z,Y}$};
\draw[thick] (Z) -- (X) node [midway,left=5pt,color=black!55!green] {$+\delta_i : Z=Y_i$};
\draw (Z) -- (I);
\draw (Z) -- (J);
\draw (Z) -- (J');
\draw (Z) -- (I');
\draw (I) -- (IJ);
\draw (J) -- (IJ);
\draw (I') -- (IJ');
\draw (J') -- (IJ');

\node[circle,draw] at (14,5) (Z1) {$Z$};
\node[] at (13,8) (Y1) {$Y_i$};
\node[] at (13,2) (X1) {$X$};
\node[] at (16,5) (J1) {$J$};
\node[] at (15,3) (IJ1) {$Z\cap J$};
\node[] at (15,7) (IJ2) {$Z\cup J$};

\node[color=black!10!red] at (15,5) (sij1) {$-\sigma_{Z,J}$};

\draw[thick] (Z1) -- (Y1) node [midway,left=5pt,color=black!10!red] {$-\delta_i: Z=X_i$};
\draw[thick] (Z1) -- (X1) node [midway,left=5pt,color=black!10!red] {$-\mu_{X,Z}$};
\draw (Z1) -- (IJ1);
\draw (Z1) -- (IJ2);
\draw (J1) -- (IJ2);
\draw (J1) -- (IJ1);
\end{tikzpicture}
\caption{Contributions of coefficients to $\flow(Z)$}
\label{fig:flow}
\end{figure}

So the problem modeled by linear program $D$ can be interpreted as a min-cost
$s$-$t$ flow problem on a hypergraph $\mathcal L$. The vertices of $\mathcal L$
are the  subsets
of $[n]$.
Each variable $\mu_{X,Y}$ represents a directed edge with no cost and infinite capacity from $Y$ to $X$, where $X \subsetneq Y$.
The variable $\delta_i$ represents
the capacity on the directed edge from $X_i$ to $Y_i$. The cost to buy this capacity is $\delta_i \cdot c_i$.
The variable $\sigma_{X,Y}$ represents a hyperedge, and $\sigma_{X,Y}=f$ means
$f$ units of flow leave each of $X$ and $Y$, and $f$ units of flow enters each of $X \cap Y$ and $X \cup Y$. Flow
can not be created at any vertex other than the empty set.
The objective is to minimize the cost of the bought capacity subject to the constraint that
this capacity can support a unit of flow from the source $s=\emptyset$ to
the sink $t=[n]$.

For the project step of our dual-project-dual technique,
we now consider the region $\Delta$ formed by projecting the feasible region for the linear program $D$
down onto the space corresponding to the $\delta_i$ variables (and perhaps
some other variables in $D$).
Then one can form a (potentially smaller)
linear program $D'$ equivalent to $D$ by replacing the constraints in $D$ by constraints
that define $\Delta$.

The final  step of the dual-project-dual technique is to take the
dual of $D'$ to obtain a linear program $P'$ that is equivalent to $P$,
but that is potentially smaller and/or simpler than $P$.

 \subsubsection{The Dual-Project-Dual Technique for Simple Instances}
In Section \ref{sect:specialcasessimple} we consider the application of
the dual-project-dual technique to simple instances.
For simple instances we show that this dual-project process results in the following linear program $D_{S}$:
\begin{align}
    D_S: && \min    & \qquad \sum_{i \in [k]} c_i \cdot \delta_i&& \label{eqn:DS} \\
         && \mbox{s.t.}  & \qquad \sum_{\substack{i \in [k]: \\  X_i \cap V  = \emptyset
                 \nonumber \\
V \cap Y_i \ne \emptyset}} \delta_i \geq 1  & \forall V \ne \emptyset, V\subseteq [n]   \nonumber
\end{align}
We then observe that the dual $P_{S}$ of the linear program $D_S$ is the natural linear program
for optimizing over weighted coverage functions:
\begin{align}
    P_S: && \max  &\qquad \sum_{\emptyset \neq V\subseteq [n]} \lambda_V && \label{eqn:PS} \\
         && \text{s.t.} & \qquad \sum_{\substack{V \ne \emptyset: \\  X_i \cap V  = \emptyset \\ V \cap Y_i \ne \emptyset}} \lambda_V \le c_i  & \forall i \in [k] \nonumber
\end{align}
 Here $\lambda_V$ is the weight of the vertex connected to vertices in $V$
 in the standard bipartite representation of a weighted coverage function (without loss of generality one can assume
 that there is only one such vertex).
Thus we rederive the fact that $DC[N_n] = DC[\Gamma_n^*] = DC[\Gamma_n]$
for  simple instances.

We then  consider the separation problem for $D_S$,
where the values of $\delta_i$'s are given,
and the goal is to determine whether this is $D_S$-feasible.
We show that this separation problem is equivalent to $n$ related
$s$-$t$ flow problems on a particular subgraph $G$ of $\mathcal L$.
The vertices in $G$ are the empty set, the singleton sets,
and the sets $Y_i$, $i \in [k]$.
The edges in $G$ are all the $\delta_i$ edges from $\mathcal L$,
and all the $\mu_{X,Y}$ edges from $\mathcal L$ where $X$ and $Y$
are vertices in $G$ and $|X| \le 1$.
We then show that a setting of the $\delta_i$ variables is feasible in
$D_S$ if and only if for all $t \in [n]$ it is possible to route a unit
of flow in $G$ from a source $s=\emptyset$ to the sink $\{t\}$.
Note that each of these flow problems is independent, so while each flow
has to respect the edge capacities, the aggregate flows over all sinks can
exceed the edge capacities. As an immediate consequence, we can conclude
that there is a polynomial-time combinatorial algorithm for the separation
problem for $D_s$, and the following  linear program
$D'_S$ is equivalent to $D_S$:
\begin{align}
    D'_S: & &\min \sum_{i \in [k]} c_i \cdot \delta_i \label{eqn:DprimeS} \\
          & \textrm{s.t.} &\qquad f_{i,t} &\leq \delta_i  && \forall i \in [k] & \forall t
          \in [n] \nonumber\\
          && \qquad \flow_t(t) &\geq 1  & &&  \forall t \in [n]\nonumber\\
          && \qquad \flow_t(\emptyset) &\geq -1  &&& \forall t \in [n]\nonumber\\
          && \qquad \flow_t(Z) &\geq 0  && \forall Z   \in  G \setminus \{\emptyset\}
    \setminus \{t\} & \forall t \in [n] \nonumber
\end{align}
where $ \flow_t(Z)$ is  defined as follows:
\begin{multline}
   \flow_t(Z) :=
     \sum_{i: Z = Y_i} f_{i,t} -
   \sum_{i: Z=X_i }f_{i,t} +
   \sum_{X: X\subset Z}\mu_{X,Z, t}+ \sum_{Y: Z\subset Y}\mu_{Z,Y, t}
 \nonumber
\end{multline}
Here the interpretation of $f_{i,t}$ is the flow from $X_i$
to $Y_i$ in $G$ for the flow problem where
the sink is $\{t\}$.
As $D'_S$ is of polynomial size, this yields a polynomial-time algorithm
to compute the coverage/entropic/submodular bound for simple instances.

 \subsubsection{The Dual-Project-Dual Technique for Strongly Connected Components}
In section \ref{sect:specialcasesSCC} we  consider applying the dual-project-dual
technique for computing the polymatroid bound for
general difference constraints.
Further we consider the effect of the strongly connected components $\CC_1, \ldots \CC_h$ of
the dependency graph for the difference constraints.
In this case we consider the projection from the feasible region for $D$  onto the
space spanned by the $\delta_i$ variables, the $\mu_{X,Y} $ variables where $Y$ is a subset of some connected component of
the dependency graph, and the $\sigma_{X, Y}$ variables where $X \cup Y$ is a subset of some connected component of
the dependency graph.
Without loss of generality, assume  $\CC_1, \ldots, \CC_h$ is a topological sort of
 the strongly connected components of the dependency graph.
We show that the result of such a projection is the following
linear program $D_{SCC}$:

\begin{equation}
\begin{array}{rrclcl}
\displaystyle D_{SCC}: \quad \quad \min & \multicolumn{3}{l}{ \sum_{i \in [k]} c_i \cdot \delta_i } \\
\textrm{s.t.} & \flow(Z,\CC_j) &\ge & 0    \qquad \forall j \in [h] \quad \forall Z~\emptyset \subsetneq Z \subsetneq \CC_j \\
&\flow(\emptyset, \CC_j)  & \geq & -1 \qquad \forall j \in [h]\\
&\flow(\CC_j,\CC_j)
     &\ge& 1  \qquad \forall j \in [h]
\end{array} \nonumber
\end{equation}
where we define $\flow$ as follows:
\begin{align*}
    \flow(Z,U)&:= &
  \sum_{\substack{ I \cup J \subseteq U\\   {I \perp J } \\{I\cap J= Z}}}
  \sigma_{I,J} + \sum_{\substack{{I \perp J } \\
    {I\cup J= Z}}} \sigma_{I,J} - \sum_{\substack{ J \in U \\  J \perp Z}} \sigma_{Z,J} \\
    &&-\sum_{ X\subsetneq Z} \mu_{X,Z} + \sum_{ Z\subsetneq  Y \subseteq U} \mu_{Z,Y}
   +\sum_{\substack{ i \in [k] \\  Z= Y_i \cap U}} \delta_i - \sum_{\substack{ i \in [k] \\  Z= X_i \cap U}} \delta_i  \\
\end{align*}
The way to think about the problem modeled by
the linear program $D_{SCC}$ is that consists
of essentially one hypergraph flow problem, as is modeled by the linear program $D$, for each connected
component. Moreover, the flow problems for the connected components
are independent with the exception
of sharing the capacities of the projection of the
difference constraints into the connected components.

The  dual of the linear program $D_{SCC}$  is the following linear program $P_{SCC}$:
 \begin{equation}
\begin{array}{rrclcl}
\displaystyle P_{SCC}: \quad \max & \multicolumn{3}{l}{\sum_{j=1}^j (h_j(\CC_j) - h_j(\emptyset))} \\
\textrm{s.t.} & h_j(Y\cup X) - h_j(X) - h_j(Y) +  h_j( Y \cap X) & \leq & 0 \qquad \substack{\forall j \in [h] \\ \forall X  \forall Y  X \perp Y \\ \text{ and } X \cup Y \in \CC_j} \\
&h_j(Y) - h_j( X) & \geq & 0 \qquad \substack{ \forall j \in [h] \\ \forall X \forall Y  X \subsetneq Y  \\ \text{ and } Y \in \CC_j}\\
 &\displaystyle \sum_{j=1}^j (h_j(\CC_j \cap Y_i) - h_j(\CC_j \cap X_i))   &\le& c_i  \qquad \forall i \in [k]
\end{array} \nonumber
\end{equation}
The linear program $P_{SCC}$ models computing the optimal objective over what we call semimodular functions.
A function is semimodular with respect to a partition $\CC_1, \ldots \CC_h$ of the
universe if for each $j \in [h]$ there exists a polymatroid function
function $h_j$ on $\CC_j$ such that for all $X \subseteq [n]$, it is the case that $h(X) = \sum_{j \in [h]} h_j(X \cap \CC_j)$.
So if the strongly connected components are singletons, then a semimodular function is modular.
Thus we can conclude that there is always an optimal solution to $DC[\Gamma_n]$ that is semimodular
with respect to the connected components,
and we recover the result that the modular bound, the entropic bound and
the polymatroid bound are equal for acyclic instances.
As the size of $P_{SCC}$ is bounded
by a polynomial function in $n$ times an exponential function of the maximum number of vertices in any strongly connected component,
this yields a fixed parameter tractable algorithm for computing the polymatroid bound
when the parameter is the maximum number of vertices in any strongly connected component.

Note that given these results one can view the conversion of an optimal polymatroid function to an optimal
modular function for acyclic instances in \cite{DBLP:conf/pods/000118}, and the conversion of an optimal polymatroid function into
an optimal weighted coverage function for simple instances in \cite{DBLP:conf/pods/KhamisK0S20}, as being
equivalent to the dual process of projecting down to the
variables in the objective in the dual space.

 \subsubsection{Lower Bound Reductions}
Whether the polymatroid bound for
arbitrary difference constraints can be computed in polynomial-time is a fascinating, and seemingly challenging, open question.
It is also natural to ask whether we can apply our techniques to
other natural classes of difference constraints,
but even this is challenging.  We have a collection of results that
illustrate some of the obstacles to extending our results. We show that computing the coverage bound for general
difference constraints is NP-hard.
We show how to efficiently reduce the problem
of computing the polymatroid bound on
general difference constraints to computing the polymatroid bound on difference constraints  that are
 a union of an acyclic instance and a simple instance.
 And we show how to efficiently reduce the problem
of computing the polymatroid bound on
general difference constraints to computing the polymatroid bound on difference constraints  where for all difference constraints $i\in [k]$
it is the case that $|X_i| \le 2$ and $|Y_i| \leq 3$.
This shows that computing the polymatroid bound for such instances is as hard as computing
the polymatroid bound in general. 

Whether the polymatroid bound for
arbitrary difference constraints can be computed in polynomial-time is a fascinating, and seemingly challenging, open question.
It is also natural to ask whether we can apply our techniques to
other natural classes of difference constraints,
but even this is challenging.  We have a collection of results that
illustrate some of the obstacles to extending our results.  
In section     \ref{sect:normalhard} we show that computing the coverage bound for general
difference constraints is NP-hard.
In section \ref{sect:polymatroidreductions}  we show how to efficiently reduce the problem
of computing the polymatroid bound on
general difference constraints to computing the polymatroid bound on difference constraints  that are
 a union of an acyclic instance and a simple instance.
 And we show how to efficiently reduce the problem
of computing the polymatroid bound on
general difference constraints to computing the polymatroid bound on difference constraints  where for all difference constraints $i\in [k]$
it is the case that $|X_i| \le 2$ and $|Y_i| \leq 3$.
This shows that computing the polymatroid bound for such instances is as hard as computing
the polymatroid bound in general.

\section{Formal Definitions}

A function $h:2^n \rightarrow \Re^+$ is {\em entropic} if there exist
discrete\footnote{It is beyond the scope of this discussion to deal with subtleties arising
from differential entropies in the continuous case.} random variables $z_1, \ldots, z_n$
such that for all $X \subseteq [n]$ it is the case
that $h(X)$ is  the entropy of the marginal distribution on the variables $z_j$ where $j \in X$.
A function $h:2^n \rightarrow \Re^+$ is a weighted coverage function if there exists a positive integer $m$,
a collection of subsets $T_1, \ldots T_n$ of $[m]$, and nonnegative weights $w_1, \ldots, w_m$
such that $h(X) = \sum_{j \in [m]: \exists i \in X  j \in T_i} w_j$.
A function $h:2^n \rightarrow \Re^+$ is a modular function if there exists
nonnegative numbers $z_1, \ldots z_n$ such that $h(X) = \sum_{i \in X} z_i$.
A function $h:2^n \rightarrow \Re^+$ is a polymatroid  function if
it is nonnegative, normalized ($h(\emptyset) =0$), monotonitically nondecreasing ($h(X) \le h(Y)$ if $X \subseteq Y$)
and submodular ($h(Y\cup X) +h (Y\cap X) \leq h(X) + h(Y)$).
An algorithm $A$  is a fixed parameter tractable algorithm in the parameter $k$
if the running time of $A$ can be bounded by a polynomial in the input size times
some function of $k$.

\section{The Dual Project Dual Approach for Simple Instances}
\label{sect:specialcasessimple}

We  show in Lemma \ref{lem:simplecut} that the feasible region of  the linear program
$D_S$  is identical to the feasible region of $D'_S$.
We then show in Lemma \ref{lem:simpleliftA}
and Lemma \ref{lem:simpleliftB}
that for simple instances
the linear program $D$ is equivalent to
the linear program $D'_S$.

\begin{lemma} \label{lem:simplecut}
The feasible region of  the linear program $D_S$
is identical to the feasible region of $D'_S$.
\end{lemma}

\begin{proof}
Assume that for some setting of the $\delta_i$
variables, that $D'_S$ is infeasible.
Then there exists a $t \in [n]$ such that the
max flow between $s=\emptyset$ and $\{t\}$ is less
than 1. Since the value of the maximum  $s$-$t$ flow
is equal the value of the minimum $s$-$t$ cut,
there must be a subset $C$ of vertices in
$G$ such that $s \notin C$ and $t \in C$,
where the aggregate capacities entering $C$ is
less than one. Thus by taking $V := \{ i \in [n] \ | \ \{i\} \in C \}$ 
we obtain a violated constraint for $D_S$.

Conversely, assume that for some setting of the $\delta_i$
variables, that $D_S$ is infeasible.
Then there is a nonempty $V $ such that
$ \sum_{\substack{i \in [k]: \\  X_i \cap V  = \emptyset \\ V \cap Y_i \ne \emptyset}} \delta_i < 1  $. Consider the cut 
$(V(G) \setminus W, W)$, where $W := \{ \{i\} \ | \ i \in V\}$.  This cut has value less than one.

 Thus again by  appealing to
the fact that the value of the minimum $s$-$t$ cut
is equal to the maximum $s$-$t$ flow. we can conclude
that this setting of the $\delta_i$ variables is
not feasible for $D'_S$.
\hfill $\square$
\end{proof}

\begin{lemma} \label{lem:simpleliftA}
For simple instances,
if a setting of the $\delta_i$ variables
can be extended to a feasible
solution for
the linear program $D$
then this same setting of the $\delta$ variables
is feasible for the linear program $D'_S$.
\end{lemma}
\begin{proof}
We prove the contrapositive. Consider a setting
of the $\delta_i$ variables that is not
feasible for $D'_S$. Then we know that
there exists a $t \in [n] $ such that there is a cut of value less than one that
that separates  $s=\emptyset$ and $\{t\}$   in $G$.

Let $V$  be the union of all singleton sets that are on the same side of this cut as $\{t\}$. We know that, $\sum_{i \in [k], X_i \cap V = \emptyset, Y_i\cap V \neq \emptyset}\delta_i < 1$.

Now consider the hypergraph $\mathcal L$ that is the lattice as defined in Section~\ref{sec:contirbutions}.  Let $V'$ be all nodes in $\mathcal L$ which contain at least one vertex in $V$ and $C$ be the remaining verticies. Then in the hypergraph $\mathcal L$ the aggregate flow
into vertices in $V'$ (i.e. out of $C$) can be at most
$\sum_{i \in [k], X_i \cap V = \emptyset, Y_i\cap V \neq \emptyset}\delta_i < 1$.
This is because no $\mu_{X,Y}$ edge can
cause flow to enter $V'$ from $C$;
and  no $\sigma_{X,Y}$ hyperedge can cause flow to leave $C$
as if $X \in C$ and $Y\in C$ then $X \cup Y \in C$,
and if either $X \notin C$ or $Y \notin C$ then $\sigma_{X,Y}$ does not route
any net flow out of $C$.
\hfill $\square$
\end{proof}

\begin{lemma} \label{lem:simpleliftB}
For simple instances, if
a setting of the $\delta_i$ variables is feasible  for
the linear program $D'_S$
then this same setting of the $\delta$ variables
can be extended to a feasible solution for the linear program  $D$.
\end{lemma}

The rest of the section is devoted to proving Lemma~\ref{lem:simpleliftB}.
We constructively show how to extend a feasible solution for $D'_S$ to a feasible solution for $D$ by
setting
$\mu_{X,Y}$ and $\sigma_{X,Y}$ variables.
Let $q$ be an integer such that the setting of every $\delta_i$, $i \in [k]$, variable is an integer
multiple of $1/q$, and let $\epsilon = 1/q$.
Initially $\mu_{ X_j,  Y_j} = \delta_{j}$ for each difference constraint $j \in [k]$,
and all other $\mu $ and $\sigma$ variables are zero.
We now give an iterative process to modify these variable settings.
The outer loop of the constructive algorithm to iterates over  $i\in[n]$.
This  loop will maintain the following outer loop invariant on the setting of the variables in $D$:
\begin{enumerate}
\item The excess  at the vertex $[i]$ in $D$ is $1$.
\item
The excess at every vertex in $D$, besides $\emptyset$ and $[i]$ is zero.
    \item For every  $j\in [k]$, if  $Y_j \cup [i] \neq X_j \cup [i]$ then  $\sum_{0 \leq t \leq i}   \mu_{[t] \cup X_j, [t] \cup Y_j} = \delta_j$.
    \item Each variable is an integer multiple of $\epsilon$.
\end{enumerate}

Note this inductive invariant is initially satisfied, if one interprets $[0]$ to be the empty set,
and will represent a feasible solution for $D$ when $i=n$.
To extend the inductive hypothesis from $i$ to $i+1$, let
 $\mathcal{P}^{i+1}  $ be the collection of simple flow
  paths in the graph $G$ that each route an $\epsilon$ unit of flow from $s=\emptyset$ to $t=\{i+1\}$ in $D'_S$.
  Our construction then iterates through the paths in $\mathcal{P}^{i+1}$, which we call the forward path process, and then iterates through these paths again in what we call the restorative process.

The forward path process processes 
 a path  $P$ in $\mathcal{P}^{i+1}$ with edges $(A_1, B_1), \ldots, (A_u, B_u)$,  where $A_1 = \emptyset$ and $B_u =\{i+1\}$ as follows.
Let $P[\ell]$  be the path  in $\mathcal L$
  that is formed from $P$ by deleting  edges $(A_h, B_h)$  where
  $A_h \cup [\ell] = B_h \cup [\ell]$, and replacing edges $(A_h, B_h)$ where
  $A_h \cup [\ell] \ne B_h \cup [\ell]$ by the edge $(A_h \cup [\ell], B_h \cup [\ell])$.
 Let $\mathcal{P}^{i+1}[\ell]$ be the collection of all such $P[\ell]$ over all $P \in \mathcal{P}^{i+1}$.
  Our construction then iterates through the edges $(A \cup [i], B \cup [i])$ in $P[i]$
  from $[i]$ to $[i+1]$ (where $(A, B)$ is the corresponding edge in $P$), processing each edge as follows:
\begin{enumerate}
\item  \label{step:supset}  If $B \subset A$ then increase $\mu_{B \cup [i], A \cup [i]}$ by $\epsilon$.
\item \label{step:subset} Else:
Let $t\le i$ be such that  $\mu_{A \cup [t], B \cup [t] } \geq \epsilon$.
\begin{enumerate}
\item If $A \cup [i] \subset B \cup [t]$ then   decrease $\mu_{A \cup [t], B \cup [t] }$ by $\epsilon$.
\item Else:
\begin{enumerate}
    \item
Decrease  $\mu_{A \cup [t], B \cup [t] }$ by $\epsilon$.
\item
 Increase $\sigma_{B \cup [t],A \cup [i]}$ by $\epsilon$.
 \item
If $(A \cup [i]) \cap (B \cup [t]) \neq  A \cup [t]$ then increase $\mu_{A \cup [t], (A \cup [i]) \cap (B \cup [t])}$ by $\epsilon$.
 \end{enumerate}
\end{enumerate}
\end{enumerate}

The restorative process iterates over the paths in $\mathcal{P}^{i+1}[i+1]$,
and then iterates over the edges of each $P[i+1]$ in $ \mathcal{P}^{i+1}[i+1]$.
An edge $(A \cup [i+1], B \cup [i+1])$ in $P[i+1]$ (where $(A, B)$ is the corresponding edge in $P$) is processed as follows:
\begin{enumerate}
\item
If $B \subset A$ then
\begin{enumerate}
\item
Decrease  $\mu_{B \cup [i], A \cup [i] }$ by $\epsilon$.
\item
Increase  $\sigma_{A \cup [i],B \cup [i+1]}$  by $\epsilon$.
\end{enumerate}
\item  \label{step:supsetback}  Else increase $\mu_{A \cup [i+1], B \cup [i+1]}$ by $\epsilon$.
\end{enumerate}

   We will show in Lemma \ref{lem:edges-simple} that the outer loop of the forward path process maintains
   the following forward path loop invariant:
\begin{enumerate}
\item The excess  at the vertex $[i]$ in  $D$ is reduced by $\epsilon$ for each path processed.
\item The excess  at the vertex $[i+1]$ in $D$ is increased by $\epsilon$ for each path processed.
\item
The excess at every vertex in $D$, besides $\emptyset$ and $[i]$ and $[i+1]$  is zero.
\item For every  $j\in [k]$,
if  $Y_j \cup [i] \neq X_j \cup [i]$ then  $\sum_{0 \leq t \leq i}   \mu_{[t] \cup X_j, [t] \cup Y_j} \ge \delta_i - \epsilon \sum_{h=1}^j \mathbb{1}_j^{i+1}(h)$,
where  $\mathbb{1}_j^{i+1}(h)$ it an indicator function that is $1$ if
path $P$ contains the edge $(X_h, Y_h)$ in $G$.
Let us call the right-hand side of this constraint the remaining capacity for difference constraint $j$.
    \item Each variable is an integer multiple of $\epsilon$.
  \end{enumerate}

\begin{lemma} The  forward path process maintains the forward path loop invariant\label{lem:edges-simple}
\end{lemma}

\begin{proof}
First, note that at least one edge has to be processed on every path $P[i]$.  The first edge processed will reduce the excess at $[i]$.  The last edge processed will increase the excess at $[i+1]$.

It is then sufficient to  show that when an edge $(A \cup [i] , B \cup [i])$ in a path $Q_j^{i+1}[i]$ is processed
the excess of $A \cup [i]$ decreases by $\epsilon$,
      the excess of $B \cup [i]$ increases by $\epsilon$,  the excess of all other vertices remain unchanged, and the remaining capacity of a difference constraint $h$
      decreases by $\epsilon$
      if and only if $A  = X_h$ and $B = Y_h$.   If
$B \subset A$ then $(A, B)$ is a $\mu$ edge then what we want to prove is obvious.
So consider the case that $A \subset B$.
In this case we know  there is a difference constraint $h$ where $A=X_h$ and
$B=Y_h$, and thus $\mathbb{1}_j^{i+1}(h) =1$.
The existence of such a $t$ follows from the  loop   invariant and the fact the flow to $\{i+1\}$
in $D'_S$ uses difference constraint $h$ to an extent at most $\delta_h$.
If $A \cup [i] \subsetneq B \cup [t]$ then
this  implies $i=t$, thus the edge invariant holds and the remaining capacity for
difference constraint $h$ decreases by a most $\epsilon$.

Otherwise, note that it must be the case that  $i > t$. Further note that
$A \cup [t] \subset (A \cup [i]) \cap (B \cup [t])$. So consider how the various nodes who are effected.  \begin{itemize}
\item $A\cup [i]$:  Increasing $\sigma_{B \cup [t],A \cup [i]}$ decreases the excess from $\epsilon$ to $0$.
\item $A \cup [t]$: The excess decreases by $\epsilon$ due to the
decrease of $\mu_{A \cup [t], B\cup [t]}$ is decreased.
If
$(A \cup [i]) \cap (B \cup [t]) = A \cup [t]$ this  decrease is canceled
by the increase  of $\sigma_{B \cup [t], A \cup [i]}$, and otherwise it is canceled
by the increase of
$\mu_{A \cup [t],(A \cup [i]) \cap (B \cup [t])}$.
\item $(A \cup [i]) \cap (B \cup [t])$: The excess increases from the increase
of $\sigma_{B \cup [t], A \cup [i]}$. If
$(A \cup [i]) \cap (B \cup [t]) = A \cup [t]$ this increase is canceled
by the decrease of $\mu_{A \cup [t], B \cup [t]}$, and otherwise it is canceled
by the increase of
$\mu_{A \cup [t],(A \cup [i]) \cap (B \cup [t])}$.
\item $B \cup [t]$: The excess decreases by $\epsilon$ from the increase of   $\sigma_{B \cup [t],A \cup [i]}$ and increases by $\epsilon$ from the decrease of $\mu_{A \cup [t], B \cup [t] }$, resulting in the excess staying at zero.
\item $B \cup [i]$: Increasing $\sigma_{B \cup [t],A \cup [i]}$ increases the excess from zero to $\epsilon$.
\end{itemize} \hfill $\square$
\end{proof}

Notice that upon termination of the forward path process,
the forward path loop invariant implies that outer loop invariant is
satisfied for $i+1$, with the exception of the third invariant.  
Lemma \ref{lem:edges-simple-backward}
   shows the the restorative process makes this third invariant true,
   without affecting the other invariants.

\begin{lemma}
After the restorative  process the outer loop invariant holds for $i+1$.  \label{lem:edges-simple-backward}
\end{lemma}

\begin{proof}
First note that if an $P_j^{i+1}[i+1]$ contains no edges, then
for each difference constraint $h$ where  $(X_h, Y_h)$ is an edge in $P_j^{i+1}$
it is the case that $X_h \cup [i+1] = Y_h \cup [i+1]$, and thus
the remaining  capacity for this difference does not need to be restored.
So consider a path $P_j^{i+1}[i+1]$ that contains a positive number of edges.
Note that $P_j^{i+1}[i+1]$ is a closed loop
as both the first vertex (namely $\emptyset$) and last vertex (namely $\{ i+1\}$) in $P_j^{i+1}$ are subsets of $[i+1]$.
Thus it will be sufficient to argue that for each edge $(A \cup [i+1], B \cup [i+1]) $
in $P_j^{i+1}[i+1]$
it is the case that when when this edge is processed
the excess of $A \cup [i+1]$ increases by $\epsilon$, the excess of
$B \cup [i+1]$ decreases by $\epsilon$, the excess of all other nodes
does not change, and  if there is a difference constraint $h$ where
$A= X_h$, $B = Y_h$, and
$X_h \cup [i+1] \ne Y_h \cup [i+1]$ then the remaining capacity for
this difference constraint will increase by $\epsilon$.
If $A \subset B$ then
there is a difference constraint $h$ where
$A= X_h$, $B = Y_h$; Further, if
$X_h \cup [i+1] \ne Y_h \cup [i+1]$ then the remaining capacity for
this difference constraint will increase by $\epsilon$.
The invariants about the excesses obvious hold in this case.

So now let us consider the effects when $B \subset A$.
Note that in this case $(B \cup [i+1]) \setminus (A \cup [i]) = \{i+1\}$.
The effects on the excesses of various nodes is:
\begin{itemize}
\item $B\cup [i]$:  Its excess is decreased by $\epsilon$ due to the decrease of  $\mu_{B \cup [i], A \cup [i] }$, and its  excess is increased by $\epsilon$ due to the increase of $\sigma_{A \cup [i],B \cup [i+1]}$.
Note that the decrease of $\mu_{B \cup [i], A \cup [i] }$ here negates
the  increase of $\mu_{A \cup [i], B \cup [i] }$ when processing edge $(A \cup [i], B \cup [i])$ in $P_j^{i+1}[i]$ in the
forward path process.
\item $B \cup [i+1]$: Its excess decreases by $\epsilon$ due to the
increase of $\sigma_{A \cup [i],B \cup [i+1]}$.
\item $A \cup [i]$: Its excess is decreased by $\epsilon$
due to the increase of   $\sigma_{A \cup [i],B \cup [i+1]}$, and  increases
by $\epsilon$ due to the decrease of  $\mu_{B \cup [i], A \cup [i] }$.
\item $A \cup [i+1]$: Its excess increases by $\epsilon$ due
to the increase of $\sigma_{A \cup [i],B \cup [i+1]}$.
\end{itemize}
\hfill $\square$
\end{proof}

\section{The Dual Project Dual Approach for Strongly Connected Components}
\label{sect:specialcasesSCC}

For convenience we rewrite $D_{SCC}$ as:

\begin{equation}
\begin{array}{rrclcl}
\displaystyle \min & \multicolumn{3}{l}{ \sum_{i \in [k]} c_i \cdot \delta_i } \\
\textrm{s.t.} & D[\CC_j] & &   \qquad \forall j \in [h]
\end{array} \nonumber
\end{equation}
where $D[U]$ are the constraints
\begin{align*}
     \flow^{\sigma,\mu}(Z,U) \ge  - \flow^{\delta}(Z, U)  &&    \forall Z \quad \emptyset \subsetneq Z \subsetneq U \  \\
  \flow^{\sigma,\mu}(U,U)
     \ge 1 - \flow^{\delta}(U,U)& &  \\
\end{align*}
where we define $\flow^{\sigma,\mu}$ as follows:
\begin{align*}
    \flow^{\sigma,\mu} (Z,U)&:= &
  \sum_{\substack{ I \cup J \subseteq U\\   {I \perp J } \\{I\cap J= Z}}}
  \sigma_{I,J} + \sum_{\substack{{I \perp J } \\
    {I\cup J= Z}}} \sigma_{I,J} - \sum_{\substack{ J \in U \\  J \perp Z}} \sigma_{Z,J} \\
    &&-\sum_{ X\subsetneq Z} \mu_{X,Z} + \sum_{ Z\subsetneq  Y \subseteq U} \mu_{Z,Y} \\
\end{align*}
and $\flow^{\delta}$ as follows:
\begin{align*}
    \flow^{\delta}(Z,U)&:= &\sum_{\substack{ i \in [k] \\  Z= Y_i \cap U}} \delta_i - \sum_{\substack{ i \in [k] \\  Z= X_i \cap U}} \delta_i  \\
\end{align*}

\begin{lemma}
A setting of the $\delta$ variables in the linear program
$D$ can be extended to a feasible solution for $D$
if and only if this same setting of the $\delta$ variables
can be extended to a feasible solution in $D_{SCC}$.
\end{lemma}

\begin{proof}
Assume that a setting of the $\delta$ variables is not
feasible for $D_{SCC}$.
Then there must exist a $j$ such
that $D[\CC_j]$ is infeasible.
By Farkas' lemma, if $D[\CC_j]$
is not feasible, then there is polymatroid function
$F_j$
on the lattice of subsets of $\CC_j$ such that
$$\sum_{Z \subseteq \CC_j}  \flow^{\delta}(Z,\CC_j) \cdot F_{j}(Z) < F_j(\CC_j)$$
We now use $F_{j}$ to define a poly-matroid
function $F$ on
the full lattice as follows:
$$ F(Z) = F_{j}(Z \cap \CC_j)$$
Note that then
$\sum_{Z \subseteq [n]}  \flow^{\delta}(Z,\CC_j) F(Z)< F([n])$. 
Thus by Farkas' lemma $D$ must be infeasible.

Fix a collection of variables $\delta$, $\mu^{SCC}$ and $\sigma^{SCC}$ variables that are feasible for $D_{SCC}$. We want to show that
$\delta$ is feasible for $D$ by setting $\sigma$ and $\mu$ appropriately.
The proof is by induction on the number of connected
components of the dependency graph.
The induction invariant is that after $j$ iterations
a flow of one has been routed to the set
$\cup_{i=1}^j \CC_i$ in the lattice.
Initially, set all $\sigma$ and $\mu$ variables to $0$.
Say that a unit flow has reached $ \CC^* := \cup_{i=1}^{j-1}\CC_i$ inductively.  Consider iteration $j$.
We set the now variables as follows.

For each variable $\sigma^{SCC}_{X,Y}$ in $D[\CC_j]$ where $X \subsetneq Y \subseteq \CC_j$, set  the variable $\sigma^D_{X \cup \CC^*,Y \cup \CC^*}$ in $D$ to $\sigma^{SCC}_{X,Y}$.
  For each variable $\mu^{SCC}_{X,Y}$ in $D[\CC_j]$ where $X, Y \subseteq \CC_j$ and $X \perp Y$,  set  the variable $\mu^D_{X \cup \CC^*,Y \cup \CC^*}$ in $D$ to $ \mu^{SCC}_{X,Y}$.
The value of the  $\delta_{i}$ are the same in $D$ and $D_{SCC}$.

For each  $i \in [k]$ let $m(i)$ be a real number
such that for strongly connected components $\CC_j $ with $j < m(i)$
it is the case that $\CC_j \cap (Y_i \setminus X_i) = \emptyset$,
for strongly connected components $\CC_j $ with $j > m(i)$
it is the case that $\CC_j \cap X_i = \emptyset$, and if
$m(i)$ is an integer then it is the case that
 $\CC_{m(i)} \cap Y_i \setminus X_i \ne \emptyset$ and
 $\CC_{m(i)} \cap X_i \ne \emptyset$.
 For each $j > m(i)$ such that $(Y_i \setminus X_i) \cap \CC_j \ne \emptyset$ then
 $\sigma^{D}_{\CC^*, Y_i \cap (\CC^* \cup \CC_j)} = \delta_i$.
 This pushes $\delta_i$ units of flow from $\CC^*$ to $\CC^* \cup (Y_i \cap \CC_j)$ in $D$,
 essentially replacing the $\delta_i$ flow in $D[\CC_j]$, and
 pushes $\delta_i$ units of flow from $Y_i \cap (\CC^* \cup \CC_j)$ to $\CC^* \cap Y_i$ in $D$
  (call this a down push).
And if $ m(i)$ is an integer then
 $\sigma^{D}_{\CC^*  \cup (\CC_j \cap X_i), Y_i \cap (\CC^* \cup \CC_j)} = \delta_i$.
  This pushes $\delta_i$ units of flow from $\CC^* \cup (\CC_j \cap X_i)$ to $\CC^* \cup (Y_i \cap \CC_j)$ in $D$,
 essentially replacing the $\delta_i$ flow in $D[\CC_j]$,  and
 pushes $\delta_i$ units of flow from $Y_i \cap (\CC^* \cup \CC_j)$ to $(\CC^* \cap X_i) \cup  (Y_i \cap \CC_j)$ in $D$
  (call this a down push).

 For each difference constraint $i \in [k]$ we set some $\mu$ variables in $D$ as follows.
Let the down pushes of flow  in $D$
  involving difference constraint $i$ that have been constructed so far be: $B_\ell$ to $A_\ell$,
  $B_{\ell-1}$ to $A_{\ell-1}$, $\ldots$, $B_1$ to $A_1$
  such that
  $$X_i \subseteq A_1 \subseteq B_1 \subseteq A_2 \subseteq B_2 \subseteq \ldots \subseteq A_\ell \subseteq B_\ell = Y_i $$
We then connect these down pushes up by setting $\mu^D_{X_1, A_1} = \delta_i$ if $X_1 \ne A_1$,
and setting each $\mu^D_{B_i, A_{i+1}} = \delta_i$ if $B_i \ne A_{i+1}$ and $i \in [\ell-1]$.
The down pushes and these $\mu$ variables together route $\delta_i$ units of flow from $Y_i$ to $X_i$ in $D$.
\end{proof}

\section{Hardness of Computing Normal Bounds}
\label{sect:normalhard}

\newcommand{\po}{\Delta}
\newcommand{\cS}{\mathcal{S}}
\newcommand{\lpn}{\textsf{LP}_{\textrm{normal}}}

This section is devoted to proving the following theorem.

\begin{theorem}
    \label{thm:normal-hardness}
    The problem of maximizing the weighted coverage function value subject to difference constraints, $\DC[N_n]$, cannot be solved in polynomial time unless P = NP.
\end{theorem}

Recall that the weighted coverage bounds are obtained over functions that are linear combinations of coverage functions. To prove the theorem we will have to define several collections of difference constraints. Thus, we will directly use $(X, Y)$ or $(X, Y, c)$ to denote a difference constraint; the former hides $c$ for brevity. For a given collection of difference constraints, $G$, the LP, $D_S$, can be rewritten into the following equivalent form, $D'_S$:

\begin{align*}
   \min        &&\sum_{(X,Y)\in G} c_{X, Y} \cdot &\delta_{X, Y}
   \label{eqn:dual:lp}\\
   \text{such that} &&
    \sum_{X \subseteq W, W \not \subseteq Y, (X, Y) \in G} \delta_{X, Y} &\geq 1 \quad \forall W \subset [n],
\end{align*}
where the constraints here are equivalent to those in $D_S$ by setting $V = [n] \setminus W$.
Let $\po(G)$ denote the convex region over $\delta$ defined by the constraints in $G$. We first show the separation problem is hard.

\begin{theorem}
    \label{thm:normal-separation}
    Given a difference constraint set $G$ and a    vector $\hat \delta \in R_{\geq 0}^{|G|}$, checking if $\hat \delta \not\in \po(G)$ is NP-complete. Further, this remains the case under the extra condition that $\lambda \hat \delta  \in \po(G)$ for some $\lambda >1$.
\end{theorem}
    We prove this theorem using a reduction from the Hitting Set problem, which is  well-known to be NP-complete. In the Hitting Set problem, the input is a set of $n$ elements $E = \{e_1,\dots, e_n\}$, a collection $\cS = \{S_1, \dots, S_m\}$ of $m$ subsets of $E$, and an integer $k > 0$. The answer is true iff there exists a subset $L$ of $k$ elements such that for every set $S_i \in \cS$ is `hit' by the set $L$ chosen, i.e., $L\cap S_i \neq \emptyset$ for all $i \in [m]$.

    Consider an arbitrary instance $H$ to the Hitting Set.  To reduce the problem to the membership problem w.r.t. $\po(G)$, we create an instance for computing weighted coverage bounds that has the elements $E' = E\cup \{e^*\}$ and the following set $G$ of differece constraints and $\hat \delta$ (here we do not specify $c_{X, Y}$ associated with each difference constraint $(X, Y)$ as it can be arbitrary and we're concerned with the hardness of the membership test):
    \begin{enumerate}
        \item $(\emptyset, \{e_i\})$ for all $
         e_i \in E$ with $\hat \delta_{\emptyset, \{e_i\}} = 1/(k+1)$.
        \item $(S_i, E')$ for all $S_i \in \cS$ with $\hat \delta_{S_i, E'} = m$.
        \item $(\{e^*\}, E')$ with $\hat \delta_{\{e^*\}, E'} = m$.
    \end{enumerate}

 Let $G_1$, $G_2$, and $G_3$ denote the difference constraints defined above in each line respectively, and let $G := G_1 \cup G_2 \cup G_3$.  To establish the reduction we aim to show the following lemma.

\begin{lemma}
    There exists a hitting set of size $k$ in the original instance $H$ if and only if $\hat \delta \not \in \po(G)$.
\end{lemma}
\begin{proof}
    Let $L(W) := \sum_{(X, Y) \in G: X \subseteq W, W \not \subseteq Y} \hat \delta_{Y|X}$. Let $L_{\min} := \min_{W\subset E'} L(W)$ and $W_{\min} := \arg \min_{W \subset E'} L(W)$. To put the lemma in other words, we want to show that $H$ admits a hitting set of size $k$ if and only if $L_{\min}  < 1$.

Let $\hat \delta(G') := \sum_{(X, Y) \in G'} \hat \delta_{X, Y}$. Note that $L_{\min} \leq L(\emptyset) = \hat \delta(G_1) = \frac{m}{k+1}$. Therefore, we can have the following conclusions about $W_{\min}$.
    \begin{itemize}
        \item  $e^* \not \in W_{\min}$ since otherwise $L_{\min} \geq \hat \delta(G_3)  = m$.
        \item For all $S_i \in \cS$, $S_i \not \subseteq W_{\min}$ since otherwise $L_{\min} \geq \hat \delta(\{(S_i, E')\}) = m$.       
    \end{itemize}
    Thus, we have shown that only the difference constraints in $G_1$ can contribute to $L_{\min}$.    
    As a result,
    $$L_{\min} = L(W_{\min}) = \hat \delta ( \{(\emptyset, \{e_i\}) \in G_1:  \{e_i\} \not \subseteq W_{\min}\})  =   \frac{1}{k+1} |E \setminus W_{\min}|.$$
    As observed above, for all $S_i \in \cS$, $S_i \not \subseteq W_{\min}$, which means $(E \setminus W_{\min}) \cap S_i \neq \emptyset$. This immediately implies that $E \setminus W_{\min}$ is a hitting set.

    To recap, if $\hat \delta \not \in \po(G)$, we have $\frac{1}{k+1} |E \setminus W_{\min}| < 1$ and therefore the original instance $H$ admits a hitting set
    $E \setminus W_{\min}$ of size at most $k$.

    Conversely, if the instance $H$ admits a hitting set $E'$ of size $k$, we can show that $L(E \setminus E') = \hat \delta(\{(\emptyset, \{e_i\}) \in G_1 \; | \; e_i \in E'\}) = \frac{k}{k+1} < 1$, which means  $\hat \delta \not \in \po(G)$. This direction is essentially identical and thus is omitted. \qed
\end{proof}

The above lemma shows checking $\hat \delta \not \in \po(G)$ is NP-hard. Further, a violated constraint can be compactly represented by $W$; thus the problem is in NP. Finally, if we scale up $\hat \delta$ by a factor of $\lambda = k+1$, we show $\lambda \hat q \in \po(G)$. We consider two cases. If $E \not \subseteq W$, we have $L(W) \geq
\hat \delta( \{ i \in [n] : e_i \not \in W \}) \geq \frac{1}{k+1} \lambda = 1$. If $E \subseteq W$, it must be the case that $E = W$ since $W \neq E'$ and $E'= E \cup \{e^*\}$. In this case $L(E) = \hat \delta(G_2) \geq m \lambda \geq 1$. Thus, for all $W \subset E'$, we have $L(W) \geq 1$, meaning $\lambda \hat \delta \in \po(G)$.
 This completes the proof of Theorem~\ref{thm:normal-separation}.

\smallskip
Using this theorem, we want to show that we can't solve $D'_S$ in polynomial time unless P = NP. While there exist relationship among the optimization problem, membership problem and their variants \cite{grotschel2012geometric}, in general hardness of the membership problem doesn't necessarily imply hardness of the optimization problem. However, using the special structure of the convex body in consideration, we can show such an implication in our setting.
The following theorem would immediately imply Theorem~\ref{thm:normal-hardness}.

\begin{theorem}
    We cannot solve $\lpn$ in polynomial time unless P = NP.
\end{theorem}
\begin{proof}
Consider an instance to the membership problem consisting of $G$ and $\hat \delta$. By Theorem~\ref{thm:normal-separation}, we know checking $\hat \delta \not \in \po(G)$ is NP-complete, even when $\lambda \hat \delta  \in  \po(G)$ for some $\lambda > 1$. For the sake of contradiction, suppose we can solve $\lpn$ in polynomial time for any $w \geq 0$ over the constraints defined by the same $\po(G)$. We will draw a contradiction by showing how to exploit it to check $\hat \delta \not \in \po(G)$ in polynomial time.

Define $R := \{w \; | \; w \cdot (\delta - \hat \delta) > 0 \;\; \forall \delta \in \po(G)\}$. 
It is straightforward to see that $R$ is convex.

We claim that $\hat \delta \not \in \po(G)$ iff $R \neq \emptyset$. To show the claim suppose $\hat \delta \not \in \po(G)$. Recall from Theorem~\ref{thm:normal-separation} that there exists $\lambda > 1$ such that $\lambda \hat \delta \in \po(G)$. Let $\lambda' >0$ be the smallest $\lambda''$ such that $\lambda'' \hat \delta \in \po(G)$. Observe that $\lambda' > 1$ and $\lambda' \hat \delta$ lies on a facet of $\po(G)$, which corresponds to a hyperplane $\sum_{X \subseteq W, W \not \subseteq Y, (X, Y) \in G} \delta_{X, Y} = 1$ for some $W \subset [n]$. Let $w$ be the orthogonal binary vector of the hyperplane; so we have $w \cdot \lambda' \hat \delta = 1$. Then, $w \cdot (\delta - \lambda' \hat \delta) \geq 0$ for all $\delta \in \po(G)$. Thus, for any $\delta \in \po(G)$ we have $w \cdot  (\delta - \hat \delta) \geq (\lambda'  -1) w \cdot  \hat \delta = \frac{\lambda' - 1}{\lambda} w \cdot \lambda \hat \delta \geq \frac{\lambda' - 1}{\lambda} > 0$. The other direction is trivial to show: If $\hat \delta \in \po(G)$, no $w$ satisfies $w \cdot (\delta - \hat \delta) > 0$ when $\delta = \hat \delta$.

Thanks to the claim, we can draw a contradiction if we can test if $R = \emptyset$ in polynomial time. However, $R$ is defined on an open set which is difficult to handle. Technically, $R$ is defined by infinitely many constraints but it is easy to see that we only need to consider constraints for $\delta$ that are vertices of $\po(G)$. Further, $\po(G)$ is defined by a finite number of (more exactly at most $2^n$) constraints (one for each $W$). This implies that the following LP, \begin{align*}
    \max & \;\epsilon \\
    w \cdot (\delta - \hat \delta) &\geq  \epsilon \quad  \forall \delta \in \po(G) \\
    w &\geq 0
\end{align*}
has a strictly positive optimum value iff $R \neq \emptyset$. We solve this using the ellipsoid method. Here, the separation oracle is, given $w \geq 0$ and $\epsilon$, to determine if $w \cdot (\delta - \hat \delta) \geq \epsilon$ for all $\delta \in \po(G)$; otherwise it should find a $\delta \in \po(G)$ such that $w \cdot (\delta - \hat \delta) < \epsilon$.
In other words, we want to know $\min_{\delta \in \po(G)} w \cdot (\delta - \hat \delta)$.  If the value is no smaller than $\epsilon$, all constraints are satisfied, otherwise, we can find a violated constraint, which is given by the $\delta$ minimizing the value. But, because the oracle assumes $w \cdot \hat \delta$ is fixed, so this optimization is essentially the same as solving  $\lpn$, which can be solved by the hypothetical polynomial time algorithm we assumed to have for the sake of contradiction. Thus, we have shown that we can decide in poly time if $R$ is empty or not. \qed
\end{proof}

\section{Hard Special Cases}
\label{sect:polymatroidreductions}

In this section we present two classes of seemingly simple instances which turn out to be as hard as general instances.

\subsection{Reduction from General DCs to Acyclic DCs and Simple FDs}
\label{subsect:acyclicplussimple}

\begin{theorem}
    \label{thm:acyclicplussimple}
    For the problem of computing the polymatroid bound, an arbitrary instance can be converted into another instance in polynomial time without changing the bound, where the difference constraints (DCs) can be divided into two subsets of acyclic DCs and simple functional dependencies (FDs)---further, each FD contains exactly two elements.
\end{theorem}

\noindent
\textbf{Reduction:} Suppose we are given an arbitrary instance $I$ consisting of the universe $U := [n]$ and a set $G$ of DCs. The new instance $I'$ has $U' := \cup_{i \in [n]}\{x_i, y_i\}$ as universe where $x_i$ and $y_i$ are distinct copies of $i$ and the following set $G'$ of DCs. For each $i \in [n]$, we first add the following simple functional dependencies to $G'$:
\begin{align*}
    (\{x_i\}, \{x_i, y_i\}, 0) \\
    (\{y_i\}, \{x_i, y_i\}, 0)
\end{align*}

Then for each $(A, B, d) \in G$, we create a new DC $(A', B', d)$ by replacing each $i \in A$ with $x_i$ and each $j \in B$ with $y_j$, and add it to $G'$. By construction these DCs are from $\{x_1, x_2, \cdots, x_n\}$ to $\{y_1, y_2, \cdots, y_n\}$ and therefore are acyclic.

\smallskip
The following simple observation states that $x_i$ and $y_i$ are indistinguishable in computing the polymatroid bound for $I'$.

\begin{lemma}
    \label{lem:g-properties} Let $g$ be a submodular function that satisfies $G'$  of $I'$.
    For any $i \in [n]$ and any $B \subseteq U'$ such that $x_i, y_i \not \in B$, we have $g(B\cup \{x_1\}) = g(B \cup \{x_2\}) = g(B \cup \{x_1,x_2\}).$
\end{lemma}
\begin{proof}
By submodularity and a FD in $G'$ involving $x_i, y_i$, we have:
$$g(B\cup \{x_i,y_i\}) - g(B\cup \{x_i\}) \leq g(\{x_i,y_i\}) - g(\{x_i\}) \leq 0.$$
Mototonicity implies
$g(B\cup \{x_i,y_i\}) - g(B\cup \{x_i\}) \geq 0$ which means $f(B\cup \{x_i,y_i\}) = f(B\cup \{x_i\})$. The other equality $g(B\cup \{x_i,y_i\}) = g(B\cup \{y_i\})$ is established analogously.  \qed
\end{proof}

Henceforth, we will show the following to complete the proof of Theorem~\ref{thm:acyclicplussimple}:
\begin{enumerate}
    \item Given a monotone submodular function $f$ achieving the optimum polymatroid bound for $I$, we create a monotone submodular function $g$ for $I'$ such that $f(U) = g(U')$.
    \item Conversely, given a monotone submodular function $g$ achieving the optimum polymatroid bound for $I'$, we crate a monotone submodular function $f$ for $I$ such that $f(U) = g(U')$.
\end{enumerate}

We first show the first direction. Let $h: 2^{U'} \rightarrow 2^U$ be a set function that converts a subset of $U'$ to a subset of $U$ by counting $x_i$ and $y_i$ only once for each $i \in [n]$. Formally, $i \in h(B)$ iff $x_i \in B$ or $y_i \in B$. Then, we set $g(B) := f(h(B))$.

\begin{lemma}
    \label{lem:h-properties}
    For any $A, B \subseteq U'$, $h(A) \cup h(B) = h(A \cup B)$ and $h(A) \cap h(B) \supseteq h(A \cap B)$.
\end{lemma}
\begin{proof}
The first claim follows because  if $x_i$ or $y_i$ is in any of $A$ and $B$, it is also   in $A\cup B$. The second claim follows because if $i  \in h(A \cap B)$, we have $x_i \in A \cap B$ or $y_i \in A \cap B$ and in both cases, we have $i \in h(A) \cap h(B)$. \qed
\end{proof}

By definition we have $g(U') = f(h(U'))) = f(U)$. Therefore, we only need to show $g$ is monotone and submodular. Showing monotonicity is trivial and is left as an easy exercise. We can show that $g$ is submodular as follows. For any $A, B \subseteq U'$, we have,
\begin{align*}
    g(A) + g(B) &= f(h(A)) + f(h(B)) \geq f(h(A)\cup h(B)) + f(h(A)\cap h(B))
    \\&\geq g(h(A \cup B)) + g(h(A \cap B)) = f(A\cup B) + f(A\cap B),
\end{align*}
where the first inequality follows from $f$'s submodularity and the second from
$f$'s monotonicity and Lemma~\ref{lem:h-properties}. Thus we have shown the first direction.

To show the other direction, we define $f(A)$ to be $g(B)$ where $i \in A$ iff $x_i \in B$. By definition, we have $f(U) = g(\{x_1, x_2, \cdots, x_n\})$. Further, by repeatedly applying Lemma~\ref{lem:g-properties}, we have $g(\{x_1, x_2, \cdots, x_n\}) = g(U')$. Thus we have shown $f(U) = g(U')$. Further,  $f$ is essentially identical to $g$ restricted to $\{x_1, x_2, \ldots, x_n\}$. Thus, $f$ inherits $g$'s monotonicity and sumodularity.

This completes the proof of Theorem~\ref{thm:acyclicplussimple}.

\subsection{Reduction from General DCs to DCs $(X, Y, d)$ with $|X| \leq 2$}
\label{subsect:reducedinstances}

\newcommand{\opt}{\textsf{opt}}

\begin{theorem}
    \label{thm:2-3-reduction}
	There is a polynomial-time reduction from a general instance to an instance preserving the polymatroid bound, where for each
difference constraint $(X, Y, d)$ we have $|Y| \leq 3$ and $|X| \leq 2$. Further, the new instance satisfies the following:
\begin{itemize}
\item If $|Y| = 3$, then $|X|=2$ and $d = 0$.
\item If $|Y| = 2$, then $|X|=1$.
\end{itemize}
\end{theorem}

\paragraph{Reduction:} The high-level idea is to repeatedly replace two variables with a new variable in a difference constraint. We first discuss how to choose two variables to combine. Assume there is a difference constraint $(X, Y, d)$ where $|X| > 2$. Then we combine an arbitrary pair of elements in $X$. If $|X| \leq 2$ for all difference constraints, and there is a difference constraint where $|Y| > 3$, we combine arbitrary two variables in $Y \setminus X$.
If there is a difference constraint $(X, Y, d)$ where $|X| = 2$, $|Y| = 3$ and $d > 0$, we combine the two variables in $X$. It is important to note that we make this replacement in only one difference constraint in each iteration.

By renaming, we can assume wlog that we combine variables $n-1$ and $n$ into a new variable 0 in a difference constraint $(X, Y, d) \in G$. Then,
we create $(X', Y', d)$ and add it to  $G'$ where
$$
    (X', Y', d) :=
    \begin{cases}
        (X \setminus \{n-1, n\} \cup \{0\}, Y \setminus \{n-1, n\} \cup \{0\}, d)  &\mbox{if } \{n-1, n\} \subseteq X, Y \\
        (X, Y \setminus \{n-1, n\} \cup \{0\}, d) &\mbox{if } \{n-1, n\} \subseteq Y \setminus X
    \end{cases}
$$
Further, we add functional dependencies $(\{n-1, n\}, \{0, n-1, n\},0 )$, $(\{0\}, \{0, n-1\},0)$ and $(\{0\}, \{0, n\},0)$ to $DC'$, which we call consistency constraints. Intuitively, consistency constraints imply we have variable $0$ if and only if we have both $n-1$ and $n$. The other constrains are called non-trivial constraints.

\begin{observation}    
    If $X \neq X'$, then we have $X' = X \setminus \{n-1, n\} \cup \{0\} $ and $X = X' \setminus \{0\} \cup \{n-1, n\}.$ A similar observation holds for $Y$ and $Y'$.
\end{observation}

\paragraph{What do we have after repeatedly applying this reduction?}
Let's first see why the reduction process terminates. For a difference constraint $(X, Y, d)$, Define $c(X, Y, d) = |X| + |Y|$ for a non-trivial constraint $(X, Y, d)$. The potential is defined as the total sum of $c(X, Y, d)$ over all non-trivial constraints.  Observe that in each iteration, either $c(X, Y, d) > c(X', Y', d)$, or $|X'| = 1$ and $|Y'| = 2$.  In the latter case, the resulting non-trivial constraint $(X', Y', d)$ doesn't change in the subsequent iterations. In the former case the potential decreases. Further, initially the potential is at most $2n |G|$ and the number of non-trivial constraints never increases, where $G$ is the set of difference constraints initially given. Therefore, the reduction terminates in a polynomial number of iterations. It is now straightforward to see that we only have difference constraints of the forms that are stated in
Theorem~\ref{thm:2-3-reduction} at the end of the reduction.

\paragraph{Reduction Preserves the Polymatroid Bound.} We consider one iteration where a non-trivial constraint $(X, Y, d)$ is replaced according to the reduction described above. Let $\opt$ and $\opt'$ be the polymatroid bounds before and after performing the iteration respectively. Let $G$ and $G'$ be the sets of the difference constraints before and after the iteration respectively.

We first show $\opt \geq \opt'$. Let $g: \{0\} \cup [n] \rightarrow [0, \infty)$ be a monotone submodular function that achieves $\opt'$ subject to $G'$.  Define $f: [n] \rightarrow [0, \infty)$ such that $f(A)  = g(A)$ for all $A \subseteq [n]$. It is immediate that $f$ is montone and submodular from $g$ being monotone and submodular, as we only restricted the function to $[n]$. The following claim shows that having $0$ is equivalent to having $n-1$ and $n$ in evaluating $g$.

\begin{claim}
    \label{claim:1-2-equal}
	For any $X$, $g(X \cup \{0\} )= g(X \cup \{0, n-1, n\}) = g(X \cup \{n-1, n\})$.
\end{claim}
\begin{proof}
Due to the consistency constraints and $g$'s monotonicity, we have $g(\{0, n-1, n\}) = g(\{n-1, n\})$. Because of the consistency constraints we added and $g$'s submodularity, we have $0 \geq g(\{0, n-1\})  - g(\{0\}) \geq g(\{0, n-1, n \})  - g(\{0, n\})$.
Then due to the monotonicity, we have $g(\{0\})  = g(\{0, n-1\})$ and $g(\{0, n-1, n \}) = g(\{0, n\})$. Similarly, we can show that $g(\{0\})  = g(\{0, n\})$ and $g(\{0, n-1, n \}) = g(\{0, n-1\})$. Thus, we have shown that $g(\{0\}) = g(\{0, n-1, n\})$.

The first equality in the claim follows since
$0 = g(\{0, n-1, n\}) - g(\{0\}) \geq g(X \cup \{0, n-1, n\}) - g(X \cup \{0\}) \geq 0$. The second equality can be shown similarly.   \qed
\end{proof}

We now check if $f$ satisfies $G$. Because we only replaced
$(X, Y, d) \in G$, we only need to show that $f$ satisfies it.
 We need to consider two case:
\begin{itemize}
	\item When $\{n-1, n\}  \subseteq X \subseteq Y$. Then, we have $g( Y \cup \{0\} \setminus \{n-1, n\}) - g(X \cup \{0\} \setminus \{n-1, n\}) \leq d$. By Claim~\ref{claim:1-2-equal}, we have $g(Y) - g(X) = g( Y \cup \{n-1, n\}) - g(X \cup \{n-1, n\}) \leq d$. By definition of $f$, we have $f(Y) - f(X) \leq d$.
	\item When $\{n-1, n\}  \subseteq Y \setminus X$. In this case, $X' = X$ and $Y' = Y \cup \{0\} \setminus \{n-1, n\}$; thus we have $g(Y \cup \{0\} \setminus \{n-1, n\}) - g(X) \leq d$. Thanks to Claim~\ref{claim:1-2-equal} and $f$'s definition, we have $f(Y) - f(X) \leq d$, as desired.
\end{itemize}

Finally, $f([n]) = g([n]) = g([n] \cup \{0\}) = \opt'$ due to Claim~\ref{claim:1-2-equal}. Since we have shown $f$ is a feasible solution for $G$, we have $\opt \geq f([n])$. Thus, we have $\opt \geq \opt'$ as desired.

\medskip
We now show $\opt \leq \opt'$. Given $f$ that achieves $\opt$ subject to $G$, we construct $g: \{0\} \cup [n] \rightarrow [0, \infty)$ as follows:
\begin{equation}
	g(A) :=
	\begin{cases}
		f(A) & \mbox{if } 0 \not \in A \\
		f(A \setminus \{0\} \cup \{n-1, n\}) & \textnormal{otherwise}
 	\end{cases}
\end{equation}

We first verify that $g$ is monotone. Consider $A \subseteq B \subseteq \{0\} \cup [n]$.
If $0 \not \in A$ and $0 \not \in B$, or  $0 \in A$ and $0 \in B$, it is easy to see that is the case. So, assume $0 \not \in A$ but $0 \in B$.
By definition of $g$, it suffices show $f(A) \leq f(B \setminus \{0\} \cup \{n-1, n\})$, which follows from $f$'s monotonicity: Since $0 \not \in A$ and $A \subseteq B$, we have $A \subseteq B \setminus \{0\} \cup \{n-1, n\}$.

Secondly we show that $g$ is submodular. So, we want to show that $g(A) + g(B) \geq g(A \cup B) + g(A \cap B)$ for all $A, B \subseteq \{0\} \cup [n]$.

\begin{itemize}
	\item When $0 \not \in A$ and $0 \not \in B$. This case is trivial as $g$ will have the same value as $f$ for all subsets we're considering.
	\item When $0 \in A$ and $0 \in B$. We need to check if $f(A \setminus \{0\} \cup \{n-1, n\}) + f(B \setminus \{0\} \cup \{n-1, n\})  \geq
	f(A \cup B \setminus \{0\} \cup \{n-1, n\}) + f(A \cap B \setminus \{0\} \cup \{n-1, n\})$, which follows from $f$'s submodularity. More concretely,  we  set $A' = A \setminus \{0\} \cup \{n-1, n\}$ and $B' = B \setminus \{0\} \cup \{n-1, n\}$ and use $f(A') + f(B') \geq f(A' \cup B') + f(A' \cap B')$.
	\item When $0 \in A$ and $0 \not \in B$ (this is symmetric to $0 \not \in A$ and $0 \in B$).  We need to check if
	$$f(A \setminus \{0\} \cup \{n-1, n\}) + f(B)  \geq 	f(A \cup B \setminus \{0\} \cup \{n-1, n\}) + f(A \cap B ).$$  For $A' = A \setminus \{0\} \cup \{n-1, n\}$, we have $f(A' ) +f(B) \geq f(A' \cup B) + f(A' \cap B)$.

	So, it suffices to show
	$$	f(A' \cup B) + f(A' \cap B) \geq f(A \cup B \setminus \{0\} \cup \{n-1, n\}) + f(A \cap B )$$

Because $A' \cup B = A \cup B \setminus \{0\} \cup \{n-1, n\}$, this is equivalent to showing:
\begin{align*}
	& \ f(A' \cap B) \geq  f(A \cap B ) \\
\Leftrightarrow  &	\ f(A' \cap (B \setminus \{0\})) \geq  f((A \setminus \{0\}) \cap  (B \setminus \{0\})) \\
\Leftarrow &  \ A' \supseteq (A \setminus \{0\})  \quad \quad  \mbox{[Due to $f$'s monotonicity]}
\end{align*}
\end{itemize}

Thirdly, we show that $g$ satisfies $G'$. Suppose we replaced a non-trivial constraint $(X, Y, d)$ with $(X', Y', d)$.
 We show $g(Y') - g(X') \leq d$ by showing $f(Y) = g(Y')$ and $f(X) = g(X')$. Both cases are symmetric, so we only show $f(X') = g(X)$. If $0 \not \in X'$, then clearly we have $g(X') = f(X)$ since $X' = X$. If $0 \in X'$, then it must be the case that $X'  = X \setminus \{n-1, n\} \cup \{0\}$. By definition of $g$, we have $g(X') = f(X'  \setminus \{0\} \cup \{n-1, n\})  = f(X)$ since $X'  \setminus \{0\} \cup \{n-1, n\} = X$.

Now we also need to check $g$ satisfies the consistency constraints we created. So we show
\begin{itemize}
	\item $g(\{0, n-1, n\}) \leq g(\{0\})$. Note $g(\{0, n-1, n\}) = f(\{n-1, n\}) = g(\{0\})$ by definition of $g$. Due to $g$'s monotonicity we have already shown, we have $g(\{0, n-1\}) \leq g(\{0\})$ and $g(\{0, n\}) \leq g(\{0\})$.
	\item $g(\{0, n-1, n\}) \leq g(\{n-1, n\})$. Both sides are equal to $ f(\{n-1, n\})$ by definition of $g$.
\end{itemize}

Finally, we have $g(\{0\} \cup [n]) = f([n])$. Since $g$ is a monotone submodular function satisfying $G$, we have $\opt' \geq \opt$ as desired.

This completes the proof of Theorem~\ref{thm:2-3-reduction}.

\bibliographystyle{alpha} 
\bibliography{ipco2023}

\end{document}